\documentclass{article}

\usepackage{amsmath}
\usepackage{amsthm}
\usepackage{graphicx}
\usepackage{subfigure}

\renewcommand{\l}{\lambda}

\newcommand{\Lh}{\hat{L}}

\newcommand{\tr}{\mbox{Tr }}
\newcommand{\x}{\mathbf{x}}
\newcommand{\y}{\mathbf{y}}
\newcommand{\xh}{\hat{\mathbf{x}}}
\newcommand{\Yh}{\hat{Y}}
\newcommand{\w}{\omega}

\newtheorem{theorem}{Theorem}[section]
\newtheorem{lemma}[theorem]{Lemma}
\newtheorem{definition}[theorem]{Definition}

\begin{document}

\title{Series expansions from the corner transfer matrix renormalization group method: the hard squares model}
\author{\textbf{Yao-ban Chan} \\ \normalsize LaBRI, Universit\'e Bordeaux 1\footnote{Much of this work was done at the Department of Mathematics and Statistics, The University of Melbourne.}\\ \texttt{\normalsize chan@labri.fr}}
\date{}

\maketitle

\begin{abstract}
The corner transfer matrix renormalization group method is an efficient method for evaluating physical quantities in statistical mechanical models. It originates from Baxter's corner transfer matrix equations and method, and was developed by Nishino and Okunishi in 1996. In this paper, we review and adapt this method, previously used for numerical calculations, to derive series expansions. We use this to calculate 92 terms of the partition function of the hard squares model. We also examine the claim that the method is subexponential in the number of generated terms and briefly analyse the resulting series.
\end{abstract}


\section{Introduction}

In this paper, we adapt an efficient numerical method --- the corner transfer matrix renormalization group, or CTMRG for short --- for the purpose of generating series expansions for the properties of various statistical mechanical models.

Much work has been devoted in recent years to finding the most accurate (i.e. largest number of correct terms) series for a variety of models. One such model is the hard squares model. In this model, each spin can take the value 0 or 1, denoting an `empty' or 'occupied' site. The weight of a configuration is $z$ to the power of the number of 1 spins, and the hard squares constraint restricts the configurations so that no two 1 spins can be directly adjacent. We wish to find the partition function of the model,
\[Z_N = \sum_{\{\sigma_i\}} z^{\sum_i \sigma_i},\]
where the outside sum is over all valid configurations, and in particular the expansion in $z$ of the partition function per site:
\[\kappa = \lim_{N \rightarrow \infty} Z_N^{1/N} = 1 + z - 2z + 8z^2 - 40 z^3 + 225 z^4 - \ldots\]

The current `state of the art' for finding series expansions is the finite lattice method, pioneered by de Neef and Enting (\cite{CTM:FLM, CTM:deNeef}) and subsequently developed and refined by Jensen, Guttmann, and Enting (\cite{CTM:FLM-JGE1, CTM:FLM-JGE2}). This method takes advantage of the fact that the infinite partition function series can be approximated up to a certain order by an expression involving the partition functions of several finite lattices.

By itself, the finite-lattice approximation to $\kappa$ is not particularly effective, but the method derives its power by using transfer matrices to calculate the partition function of the finite lattices. These are matrices, denoted by $V$ and indexed by the values of a column of spins, which contain the Boltzmann weight of a column, given the spins on the sides. Multiplication by $V$ `adds' the weight of an extra column, so the partition function of the entire lattice can be built by repeating such multiplications. This can be used either to calculate the partition function of finite lattices, for the finite lattice method, or can be extended to the thermodynamic limit, where the partition function can be expressed as
\[Z = \lim_{N \rightarrow \infty} \tr V^N.\]
Therefore the partition function can be derived directly from the largest eigenvalue of the transfer matrix.

The CTMRG uses a similar concept of transfer matrices, but with a twist. Instead of containing the weight of a column, corner transfer matrices contain the weight of a full quarter of the lattice, given the spins at the boundaries. This means that the partition function is expressible as the sum of 4th powers of corner transfer matrices. The advantage of this representation is that approximating the `true' infinite-dimensional matrices by finite-size matrices is often very accurate, even for small matrices. The disadvantage is that we must evaluate the full eigenvalue spectrum of the corner transfer matrices, rather than simply the largest eigenvalue.

The CTMRG method is quite general, and can in theory be applied to any model where the Boltzmann weight of a configuration can be expressed as the product of weights of a single cell. We call such models interaction round a face (IRF) models. In terms of such a model, the face weight of the hard squares model is given by
\[\w\left(\begin{array}{cc} a & b \\ c & d \end{array}\right) = \left\{ \begin{array}{rl} 0 & \mbox{if $a=b=1$, $a=c=1$, $b=d=1$ or $c=d=1$} \\ z^{(a+b+c+d)/4} & \mbox{otherwise.} \end{array} \right.\]
We note that the factor of $1/4$ comes from the fact that each spin lies in 4 cell faces.

The original corner transfer matrix method was developed by Baxter, Enting, and various co-authors starting from 1978. In \cite{CTM:1}, Baxter developed the corner transfer matrix equations, which underpin all CTM-related methods. We discuss these equations in detail below. He also developed the corner transfer matrix method, which involved transforming the CTM equations and iterating through them until a solution was reached.

This method was applied to the Ising model for low-temperature and high-field series expansions (\cite{CTM:Ising, CTM:Ising2}), as well as the hard squares model (\cite{CTM:Hsq}). Using what would be considered today as insignificant computational power, they managed to extract a very large number of series terms for these models (in fact, we know of no longer series expansions for the hard squares model). More recently, Baxter used this method to numerically calculate the partition function for hard particle models at $z = 1$ to high precision (\cite{CTM:Hsq2}).

A notable triumph of the corner transfer matrix approach was Baxter's exact solution of the hard hexagons model (\cite{CTM:HexExact}), which he derived by noticing a pattern in the eigenvalues of the corner transfer matrices. He was then able to prove the validity of this pattern and thus solve the model.

Other applications of this method by Baxter included the 8-vertex model (\cite{CTM:8vertex1, CTM:8vertex2}), the 3d Ising model for one-dimensional matrices (\cite{CTM:3d}), and the chiral Potts model (\cite{CTM:Potts1, CTM:Potts2, CTM:Potts3}). However, in each case the actual transformations applied were model-specific, so the method remained relatively limited in application.

In 1996, Nishino and Okunishi (\cite{CTM:RG1, CTM:RG3, CTM:RG2}) used the CTM equations to derive the corner transfer matrix renormalization group method (CTMRG). This method calculates finite-size approximations to the solution of the CTM equations using a principle derived from the density matrix renormalization group method. Nishino and Okunishi applied this method numerically to various models --- the $q=5$ Potts model (\cite{CTM:RG-App1}), the 3d Ising model (\cite{CTM:RG-App4}), and the spin-$\frac{3}{2}$ Ising model (\cite{CTM:RG-App7}). They also converted it to 3-dimensional lattices in \cite{CTM:RG-App2, CTM:RG-App5, CTM:RG-App6, CTM:RG-App10}, and studied the eigenvalue distribution of the CTM matrices in \cite{CTM:RG-App8}.

In 2003, Foster and Pinettes (\cite{CTM:RG-FP1, CTM:RG-FP2}) applied this method to the self-avoiding walk model, and more recently Mangazeev \emph{et al.} (\cite{CTM:Scaling, CTM:Scaling2}) used it to evaluate the scaling function of the Ising model in a magnetic field.

All of the above applications were for numerical calculations. As far as we know, no one has tried to use this method to derive series expansions. In principle, the method can be adapted to do this with no changes. However, in practice there are some implementational difficulties, notably that we must diagonalize a matrix of series. To our knowledge this has not been attempted before. These considerations give rise to some interesting mathematics and are the subject of this paper.

It is claimed (although not actually proved) that the method has complexity $O(\alpha^{\sqrt{n}})$, as opposed to the FLM which is an exponential-time method, albeit with a small growth constant. We will briefly examine this claim.

In this paper, we apply the CTMRG method to derive series for the hard squares model. In Section \ref{sec:method}, we revise the CTM equations and the CTMRG method. In Section \ref{sec:implementation}, we discuss adjustments and algorithms needed to apply the CTMRG for series expansions. These include diagonalization of matrices of series and block eigenvalues. In Section \ref{sec:convergence} we analyse the effectiveness of our method, and analyse the resulting series in Section \ref{sec:analysis}. Finally we offer a brief conclusion in Section \ref{sec:conclusion}.

\section{The corner transfer matrix renormalization group method}\label{sec:method}

\subsection{The CTM equations}

The CTMRG method is based on Baxter's CTM equations (\cite{CTM:1}), which we restate below. In these equations, $a$, $b$, $c$, and $d$ take all possible spin values, while $A(a)$, and $F(a,b)$ are $n \times n$ matrices. $\w$ is the weight of one cell, given the spins at its corners, and $\eta$ and $\xi$ are scalars.

\begin{eqnarray}
\xi A^2(a) & = & \sum_{b} F(a,b) A^2(b) F(b,a) \label{eq:ctm1} \\
\eta A(a) F(a,b) A(b) & = & \sum_{c,d} \w\left(\begin{array}{cc} a & b \\ c & d \end{array}\right) F(a,c) A(c) F(c,d) A(d) F(d,b). \label{eq:ctm2}
\end{eqnarray}

It can be shown (for example in \cite{CTM:Thesis}) that the infinite-dimensional solution to these equations gives the partition function per site by $\kappa = \eta/\xi$. At any finite dimension, the equations are consistent and provide a lower bound (and approximation) to $\kappa$.

The CTM equations can best be understood by their graphical interpretation, thinking of the matrices as transfer matrices. This is illustrated in Figure \ref{fig:mat}. The $A$ matrices are interpreted as the transfer matrix of a quarter of a plane (or \emph{corner transfer matrix}), given the value of the spin at the corner, whereas the $F$ matrices are interpreted as a `half-row' transfer matrix, given the two spins at the end.

\begin{figure}\center
\subfigure[$A(a)$]{\includegraphics[scale=0.8]{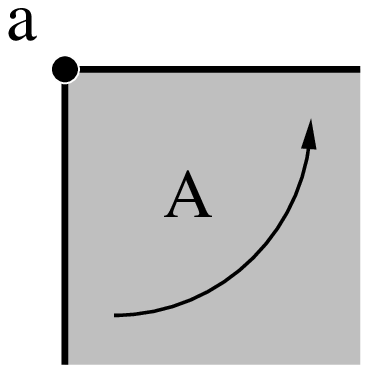}}
\subfigure[$F(a,b)$]{\includegraphics[scale=0.8]{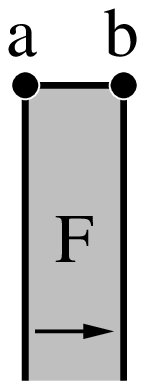}}
\caption{Graphical interpretation of the matrices in the CTM equations.}\label{fig:mat}
\end{figure}

With these interpretations of the matrices, the CTM equations can be intuitively seen to be correct, as shown in Figure \ref{fig:eqs}. Equation \ref{eq:ctm1} corresponds to adding a row onto half a plane, which multiplies the matrix by a constant factor but otherwise leaves it unchanged. Equation \ref{eq:ctm2} has a similar interpretation, but with the values of two spins fixed.

\begin{figure}\center
\subfigure[Equation \ref{eq:ctm1}]{\includegraphics[scale=0.7]{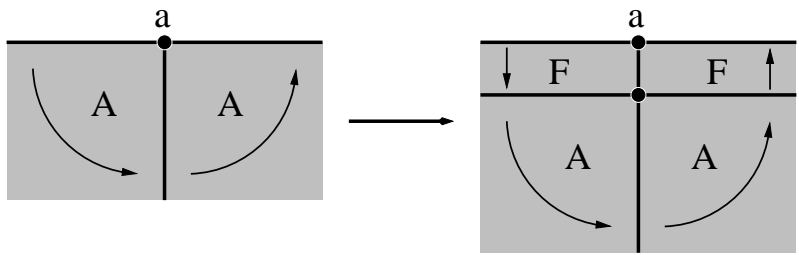}}
\subfigure[Equation \ref{eq:ctm2}]{\includegraphics[scale=0.7]{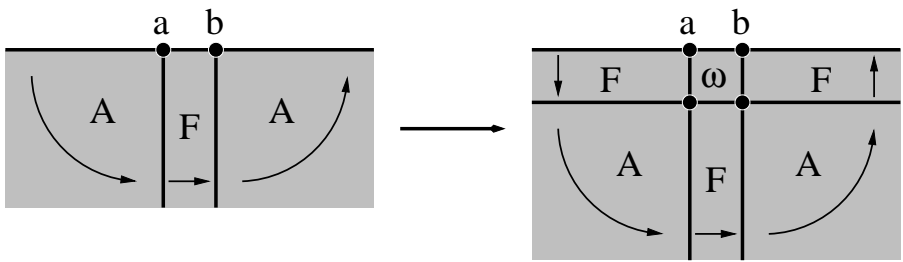}}
\caption{Graphical interpretation of the CTM equations.}\label{fig:eqs}
\end{figure}

Although this graphical interpretation suggests that each matrix be of dimension $2^p \times 2^p$ for some $p$, this is not actually necessary --- the equations hold at any size. Indeed, as we shall see later, the matrices do not all have to be of the same size for the equations to be consistent.

It is easily seen that the CTM equations are invariant under the transformations:
\begin{enumerate}
\item $A(a) \rightarrow c A(a)$;
\item $F(a,b) \rightarrow c F(a,b), \xi \rightarrow c^2 \xi, \eta \rightarrow c^2 \eta$; and
\item $A(a) \rightarrow P^T(a) A(a) P(a), F(a,b) \rightarrow P^T(a) F(a,b) P(b)$, where $P(a)$ is an orthogonal matrix of size $n \times n$.
\end{enumerate}
In particular, transformation 3 means that we can take $A(a)$ to be diagonal, with entries ordered from largest to smallest. Furthermore, transformations 1 and 2 imply that we can then take the top left entries of $A(0)$ and $F(0,0)$ to both be 1. Although this is not implied by the CTM equations, we can often take $F(a,b) = F^T(b,a)$, due to reflectional symmetry in the model.

\subsection{The renormalization group method}

The CTMRG method of Nishino and Okunishi (\cite{CTM:RG1}) calculates successive approximations to finite-size solutions of the CTM equations. The principle behind this method is that the finite-size solution maximises $\kappa$ with respect to the matrices. Since $\sum_a A^4(a)$ is the partition function of the entire plane, we wish to keep the maximum eigenvalues of the infinite-dimensional solution in the $A$ matrices. To do this we expand these matrices, and then diagonalise them. We then apply the diagonalising transformation, but keep only the largest eigenvalues, so that the matrices are shrunk back to their original size.

More specifically, given initial values for the $A$ and $F$ matrices, we expand our $A$ matrices via
\begin{equation}
A_l(a) = \left( \begin{array}{cc} \sum_b \w\left(\begin{array}{cc} a & 0 \\ 0 & b\end{array}\right) F(0,b)A(b)F(b,0) & \sum_b \w\left(\begin{array}{cc} a & 1 \\ 0 & b\end{array}\right) F(0,b)A(b)F(b,1) \\ \sum_b \w\left(\begin{array}{cc} a & 0 \\ 1 & b\end{array}\right) F(1,b)A(b)F(b,0) & \sum_b \w\left(\begin{array}{cc} a & 1 \\ 1 & b\end{array}\right) F(1,b)A(b)F(b,1) \end{array}\right). \label{eq:al}
\end{equation}

We expand the $F$ matrices in a similar fashion:
\begin{equation}
F_l(a,b) = \left( \begin{array}{cc} \w\left(\begin{array}{cc} a & b \\ 0 & 0 \end{array}\right) F(0,0) & \w\left(\begin{array}{cc} a & b \\ 0 & 1 \end{array}\right) F(0,1) \\ \w\left(\begin{array}{cc} a & b \\ 1 & 0 \end{array}\right) F(1,0) & \w\left(\begin{array}{cc} a & b \\ 1 & 1 \end{array}\right) F(1,1) \end{array} \right). \label{eq:fl}
\end{equation}

As before, these equations have graphical interpretations, shown in Figure \ref{fig:ctmrg}. We expand $A(a)$ by adding the weight of two half-rows and a single cell. This has the result of doubling the size of $A(a)$. We add the weight of a single cell to $F(a,b)$, which also doubles its size.

\begin{figure}\center
\subfigure[Expanding $A(a)$]{\includegraphics[scale=0.7]{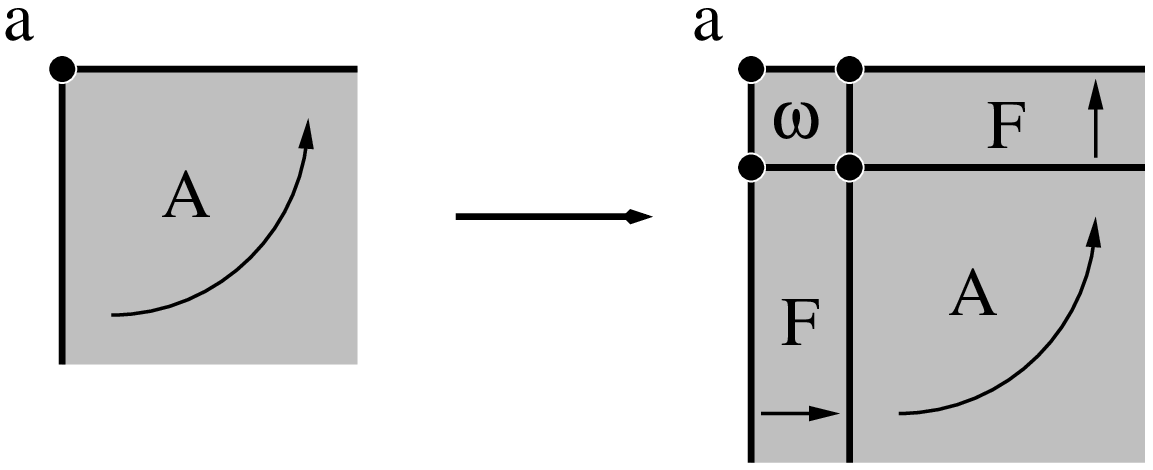}}
\subfigure[Expanding $F(a,b)$]{\includegraphics[scale=0.7]{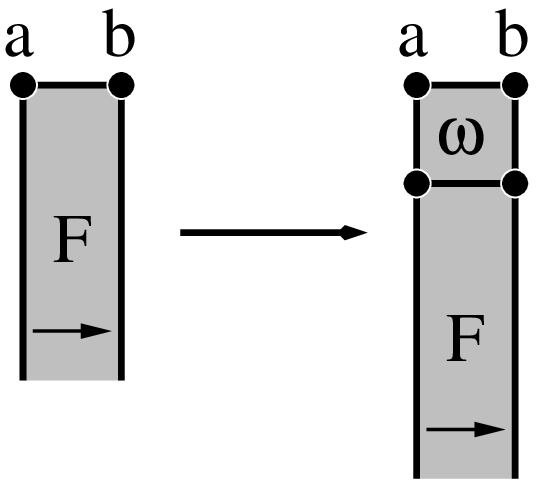}}
\caption{Graphical interpretation of matrix expansion in the CTMRG method.}\label{fig:ctmrg}
\end{figure}

Next we must reduce both these matrices. This is done by diagonalising $A_l(a)$, and then truncating the diagonalising matrix so that only the $n$ largest eigenvalues are kept. To reduce the size of the $F$ matrices, we apply a transformation consistent with transformation 3 in the previous section. If we wish to change the finite size of the solution, we can simply keep a larger number of eigenvalues of $A_l(a)$.

The CTMRG method can now be stated in full:

\begin{enumerate}
\item Start with initial approximations for $A(a)$ and $F(a,b)$.
\item Calculate $A_l(a)$ and $F_l(a,b)$ from Equations \ref{eq:al} and \ref{eq:fl}.\label{step:al}
\item Diagonalize $A_l(a)$, i.e. find orthogonal matrices $P_l(a)$ such that $P_l^T(a) A_l(a) P_l(a)$ is diagonal, with diagonal entries in order from largest to smallest.\label{step:diag}
\item If we want to expand the matrices, increase $n$.
\item Let $P(a)$ be the first $n$ columns of $P_l(a)$.
\item Set $A(a) = P^T(a) A_l(a) P(a)$ and $F(a,b) = P^T(a) F_l(a,b) P(b)$.
\item Return to step \ref{step:al}.
\end{enumerate}

Since this method does not explicitly calculate $\xi$ or $\eta$, we calculate the partition function per site $\kappa$ by the formula
\begin{equation}
\kappa = \frac{\left(\tr \sum_a A^4(a)\right)\left(\tr \sum_{a,b,c,d} \w\left(\begin{array}{cc} a & b \\ c & d \end{array}\right) A(a) F(a,c) A(c) F(c,d) A(d) F(d,b) A(b) F(b,a)\right)}{\left( \tr \sum_{a,b} A^2(a) F(a,b) A^2(b) F(b,a) \right)^2}.\label{eq:kappa}
\end{equation}
In practice, we actually use $A_l$ and $F_l$ in place of $A$ and $F$ in the above formula.

This formula also has a graphical interpretation, shown in Figure \ref{fig:kappa}. All terms are partition functions of the entire plane, but the term in the denominator contains one column more than the first term in the numerator, while the second term in the numerator contains both a row and a column more than the first term. The net effect is to isolate the partition function of a single cell.

\begin{figure}\center
\subfigure[First term in the numerator]{\includegraphics[scale=0.8]{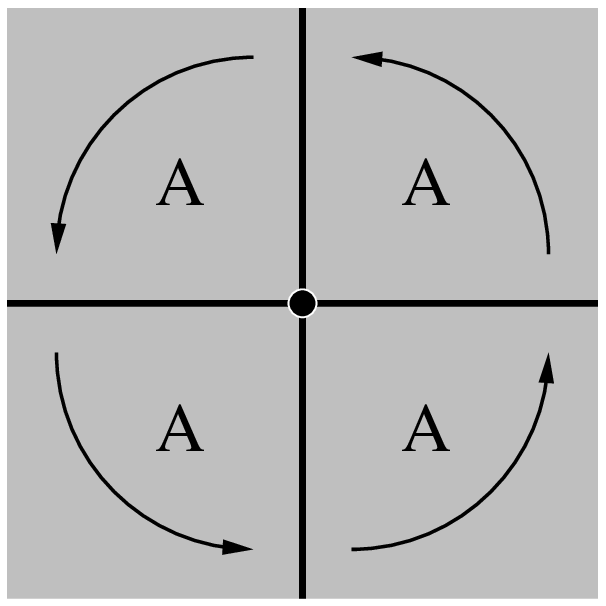}}
\subfigure[Term in the denominator]{\includegraphics[scale=0.8]{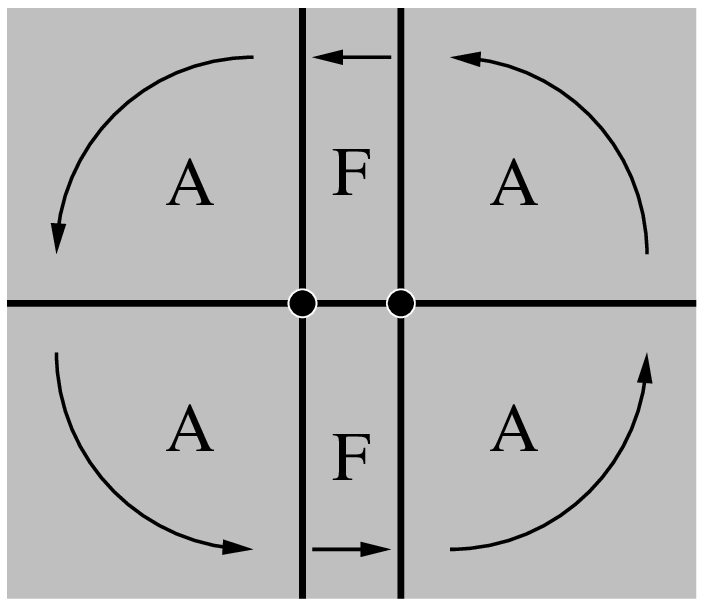}}
\subfigure[Second term in the numerator]{\includegraphics[scale=0.8]{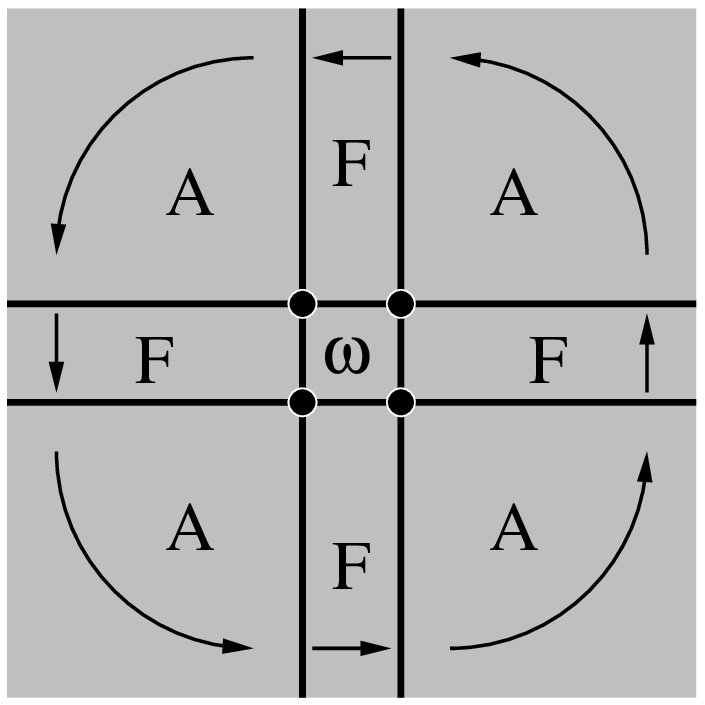}}
\caption{Terms in Equation \ref{eq:kappa}.}\label{fig:kappa}
\end{figure}

It is also easy to calculate the magnetisation via the formula
\[M = \frac{\tr \sum_a a A^4(a)}{\tr \sum_a A^4(a)}.\]

\section{Implementing CTMRG for series}\label{sec:implementation}

Although the CTMRG method has so far been only used for numerical calculations, it can work as a series-calculating tool. In this section, we discuss some difficulties that are specific to series calculations. We also mention some improvements which we have made.

\subsection{Modular arithmetic}

A standard `trick' in series calculations of long length is to perform all the calculations in integers modulo a prime. If the series that we are calculating has integer coefficients, and we know the sign of each coefficient (which is the case for the hard squares model), then we can repeat the calculations using different moduli and use the Chinese Remainder Theorem to reconstruct the original coefficients.

This difficulty of doing this in this particular algorithm arises because we need to take the square roots of some numbers (for example when normalising eigenvectors). This is partially overcome by the following lemma (taken from \cite[Section II.2]{CTM:Koblitz}).

\begin{lemma}
Let $p$ be a prime number such that $p \equiv 3\mbox{ mod }4$, and let $a$ be an integer. Then if $a$ has a square root modulo $p$, $a^{\frac{p+1}{4}}$ is a square root of $a$ modulo $p$.
\end{lemma}

This lemma is easily proved --- one such proof is in \cite{CTM:Thesis}. Unfortunately, this still does not fully solve the problem, as not all integers have square roots modulo $p$ --- this can be seen by noting that each number which has a square root has 2 distinct square roots. At small matrix sizes, this did not seem to be a problem, but at larger sizes, some of the eigenvalues of $A_l(a)$ contained square roots which were not calculable. In order to address this, we use block eigenvalues, which we describe below in Section \ref{sec:block}.

\subsection{Unequal matrix sizes}

In the CTM equations and the renormalization group method, all the $A$ and $F$ matrices are always of the same size as each other. However, this is not necessary, as the equations can be made consistent with different size matrices. Table \ref{tab:size} shows the sizes needed.

\begin{table}
\begin{tabular}{lrr}
Matrix & Size (equal sizes) & Size (unequal sizes) \\
\hline
$A(0), F(0,0)$ & $n \times n$ & $n_1 \times n_1$ \\
$A(1)$ & $n \times n$ & $n_2 \times n_2$ \\
$F(0,1)$ & $n \times n$ & $n_1 \times n_2$ \\
$F(1,1)$ & $n \times n$ & $n_2 \times n_2$ \\
$A_l(0), A_l(1), F_l(0,0), F_l(0,1)$ & $2n \times 2n$ & $(n_1+n_2) \times (n_1+n_2)$ \\
$P(0)$ & $2n \times n$ & $(n_1+n_2) \times n_1$ \\
$P(1)$ & $2n \times n$ & $(n_1+n_2) \times n_2$ \\
\end{tabular}\caption{Required sizes for the matrices.}\label{tab:size}
\end{table}

The ability to set the matrices to different sizes is useful because the number of correct series terms derived at each finite size depends on the largest eigenvalue of $A_l(a)$ that is missing from $A(a)$ (this will be discussed further in Section \ref{sec:convergence}). However, the leading powers of the eigenvalues of $A_l(1)$ increase more rapidly than those of $A_l(0)$, so we can keep $A_l(1)$ at a smaller size and still derive the same number of terms.

\subsection{Diagonalization of a matrix of series}

In step \ref{step:diag} of the CTMRG method, we diagonalize a matrix which has power series elements. Furthermore, in order to obtain series exactly to some order, the diagonalization must be exact to that order. In theory, this is impossible even for real-valued matrices. However, $A_l(a)$ often turns out to have a relatively simple structure which enables us to diagonalize it exactly.

We first used the well-known power method to calculate the eigenvalues and eigenvectors of $A_l(a)$. The following theorem justifies its use in this case. In this theorem and all other cases, we order series in lexicographical order, so that a series with a lower leading power is always considered larger than one with a higher leading power.

\begin{theorem}\label{thm:power}
Let $A$ be a symmetric $n \times n$ matrix of power series, with eigenvalues $\l_1 \geq \l_2 \geq \ldots \geq \l_n$ with leading powers $l_1, \ldots, l_n$ and corresponding eigenvectors $\x_1, \ldots, \x_n$, which are taken to have unit norm. Suppose that we have an estimate $\xh_1$ of $\x_1$ which is also of unit norm and accurate to $m$ terms, i.e.
\[\xh_1 - \x_1 = O(z^m).\]

Then
\[\|A\xh_1\| - \l_1 = O(z^{l_1+m})\]
and
\[\frac{A\xh_1}{\|A\xh_1\|} - \x_1 = O(z^{m+l_2-l_1}).\]
\end{theorem}
\begin{proof}
Write
\[\xh_1 - \x_1 = a_1 \x_1 + a_2 \x_2 + \ldots + a_n \x_n\]
where $a_i = O(z^m)$ for all $i$. Then
\begin{eqnarray*}
A \xh_1 & = & (1 + a_1) A \x_1 + a_2 A \x_2 + \ldots + a_n A \x_n \\
	& = & (1 + a_1) \l_1 \x_1 + a_2 \l_2 \x_2 + \ldots + a_n \l_n \x_n \\
	& = & (1 + a_1) \l_1 \x_1 + O(z^{l_2+m}) \\
\|A \xh_1\| & = & (1 + a_1) \l_1 + O(z^{l_2+m}) \\
	& = & \l_1 + O(z^{l_1+m}).
\end{eqnarray*}

Normalising (with some abuse of $O$-notation) gives
\begin{eqnarray*}
\frac{A\xh_1}{\|A\xh_1\|} & = & (1 + O(z^{l_2+m-l_1}))^{-1} \x_1 + \frac{O(z^{l_2+m})}{O(z^{l_1})} \\
	& = & \x_1 + O(z^{m+l_2-l_1}).
\end{eqnarray*}

\end{proof}

Theorem \ref{thm:power} shows that if the maximum eigenvalue of $A_l(a)$ is not degenerate to leading power, then every iteration of the power method produces more correct terms than the previous iteration, in both the eigenvalue and the eigenvector. If this occurs, we find the dominant eigenvalue using the power method, then deflate the matrix by normalising the eigenvector and calculating
\[A_l(a) - \lambda \x \x^T.\]
This matrix has the same eigenvalues and eigenvectors of $A_l(a)$, but with $\lambda$ replaced by 0.

While the converse of Theorem \ref{thm:power} --- if the maximum eigenvalue of the matrix is degenerate to leading power, then the power method fails --- is not always true, it sometimes holds. In these cases, we cannot use the power method. To overcome this, we shift the eigenvalues and invert. $(A - \l_0 I)^{-1}$ has the same eigenvectors as $A$, but any eigenvalue $\l$ becomes $\frac{1}{\l - \l_0}$. We use this if we know the leading terms of one of the eigenvalues of $A_l(a)$ to an order which specifies it uniquely. If we know that $\l_0$ is equal to exactly one of the eigenvalues of $A_l(a)$ up to order $z^m$, then $A_l(a) - \l_0 I$ will have one eigenvalue with leading power $z^m$, with all other eigenvalues having smaller leading power. Hence $(A_l(a) - \l_0 I)^{-1}$ will have a largest eigenvalue which is not degenerate to leading power, and we can use the power method on this matrix.

In fact, because the convergence of the power method depends on the difference in leading powers between the two largest eigenvalues, if we have a very good approximation of the required eigenvalue, shifting and inverting will result in a matrix which enables us to converge to the correct eigenvalue very quickly.

This leaves us with the problem of finding the first few terms of all the eigenvalues of $A_l(a)$ with enough precision to uniquely identify them. In practice, at small size almost all of the eigenvalues of $A_l(a)$ have distinct leading terms, if not necessarily leading powers (the first case of identical leading terms occurs at size 23 in $A_l(0)$), so it is usually sufficient to find the first term of each eigenvalue. We used various methods:

\begin{itemize}
\item The eigenvalues do not change much from iteration to iteration (of the CTMRG method). We use eigenvalues from the previous iteration as starting points, and can converge to the new eigenvalues in very few power method iterations. However, when we expand the matrices, we have one eigenvalue too few, so this is not always sufficient.
\item Sometimes, applying a few iterations of the power method does produce the leading term of the largest eigenvalue, even if that eigenvalue is degenerate to leading power. It is usually obvious when this happens, because the leading term becomes invariant within 2-3 iterations.
\item If the lowest leading power in $A_l(a)$ is $z^m$, then $[z^m]A_l(a)$ has eigenvalues which are the coefficients of $z^m$ in the eigenvalues of $A_l(a)$. Often, there will only be a few components of $A_l(a)$ with leading power $z^m$, and these will often break down into a simple block diagonal structure. If one of the blocks is of size $1 \times 1$, then that contains the leading term for an eigenvalue.
\item If one of the blocks is of size $2 \times 2$, then we can calculate leading terms for two eigenvalues by manually solving the eigenvalue equation for that $2 \times 2$ block.
\item If all else fails, we use block eigenvalues. This is described in the following section.
\end{itemize}

\subsection{Block eigenvalues}\label{sec:block}

In the diagonalization step of the CTMRG method, it is not really important to exactly diagonalize $A_l(a)$. All we need to do is to apply a similarity transformation to $A_l(a)$ which reduces it to the required size while keeping the largest eigenvalues. Certainly, diagonalizing is one way to ensure that this happens, but it is not the only way. The idea behind \emph{block eigenvalues} is that they keep the required eigenvalues, while (potentially) avoiding calculational pitfalls which may occur if we diagonalize fully. Formally:

\begin{definition}
Let $A$ be a symmetric $n \times n$ matrix. A $2 \times 2$ matrix $L$ is a \emph{block 2-eigenvalue} of $A$ with corresponding \emph{block 2-eigenvector} $Y$, where $Y$ is an $n \times 2$ matrix, if
\[AY = YL.\]
\end{definition}

Block $k$-eigenvalues (where $k > 2$) are defined in an identical manner, and all results from this section can be extended to larger $k$. The next theorem shows that block 2-eigenvalues have the property of only `representing' 2 eigenvalues.

\begin{theorem}\label{thm:block}
Let $A$ be a symmetric $n \times n$ matrix with no degenerate eigenvalues, and let $L$ be a block 2-eigenvalue of $A$ with corresponding block 2-eigenvector $Y$. Then the columns of $Y$ are linear combinations of at most two eigenvectors $\x_1$ and $\x_2$ of $A$. Furthermore, if the columns of $Y$ are linear combinations of two eigenvectors, then $L$ has eigenvalues $\l_1$ and $\l_2$, which are the eigenvalues of $A$ corresponding to $\x_1$ and $\x_2$.
\end{theorem}
\begin{proof}
Suppose that the columns of $Y$ are linear combinations of 3 eigenvectors of $A$:
\[Y = \left[\begin{array}{c|c} a_1 \x_1 + a_2 \x_2 + a_3 \x_3 & b_1 \x_1 + b_2 \x_2 + b_3 \x_3 \end{array}\right].\]
Let
\[L = \left[\begin{array}{cc} l_{11} & l_{12} \\ l_{21} & l_{22} \end{array}\right].\]

Then
\begin{eqnarray*}
AY & = & \left[\begin{array}{c|c} a_1 A \x_1 + a_2 A \x_2 + a_3 A \x_3 & b_1 A \x_1 + b_2 A \x_2 + b_3 A \x_3 \end{array}\right] \\
	& = & \left[\begin{array}{c|c} a_1 \l_1 \x_1 + a_2 \l_2 \x_2 + a_3 \l_3 \x_3 & b_1 \l_1 \x_1 + b_2 \l_2 \x_2 + b_3 \l_3 \x_3 \end{array}\right] \\
	& = & YL \\
	& = & \big[ \;\; l_{11} \left( a_1 \x_1 + a_2 \x_2 + a_3 \x_3 \right) + l_{21} \left( b_1 \x_1 + b_2 \x_2 + b_3 \x_3 \right) \\
	& & \hspace{1cm} \big| \;\; l_{12} \left( a_1 \x_1 + a_2 \x_2 + a_3 \x_3 \right) + l_{22} \left( b_1 \x_1 + b_2 \x_2 + b_3 \x_3 \right) \;\; \big] \\
	& = & \big[ \;\; (a_1 l_{11} + b_1 l_{21}) \x_1 + (a_2 l_{11} + b_2 l_{21}) \x_2 + (a_3 l_{11} + b_3 l_{21}) \x_3 \\
	& & \hspace{1cm} \big| \;\; (a_1 l_{12} + b_1 l_{22}) \x_1 + (a_2 l_{12} + b_2 l_{22}) \x_2 + (a_3 l_{12} + b_3 l_{22}) \x_3 \;\; \big].
\end{eqnarray*}

This implies that
\begin{eqnarray*}
a_1 l_{11} + b_1 l_{21} & = & a_1 \l_1 \\
a_1 l_{12} + b_1 l_{22} & = & b_1 \l_1
\end{eqnarray*}
and hence that $\left[\begin{array}{cc} a_1 & b_1 \end{array}\right]$ is a left eigenvector of $L$ with corresponding eigenvalue $\l_1$. Similarly, $\left[\begin{array}{cc} a_2 & b_2 \end{array}\right]$ and $\left[\begin{array}{cc} a_3 & b_3 \end{array}\right]$ are also left eigenvectors of $L$ with corresponding eigenvalues $\l_2$ and $\l_3$ respectively. But $L$ is a $2 \times 2$ matrix, and so can have at most 2 distinct eigenvalues, and from our assumptions, $\l_1, \l_2,$ and $\l_3$ are distinct. This is a contradiction, so the columns of $Y$ can be spanned by at most 2 eigenvectors of $A$. The second part of the theorem now follows from the observed eigenvalues of $L$.

\end{proof}

It is easy to see that the converse of this theorem is also true: any two linear combinations of two eigenvectors form a block 2-eigenvector.

We observe that block eigenvalues are not unique, even if we fix the eigenvalues of $A$ which they contain. For example, if $A$ is itself $2 \times 2$ with eigenvalues $\l_1$ and $\l_2$, then both $\mbox{diag}(\l_1, \l_2)$ and $A$ itself are block 2-eigenvalues of $A$. It is this flexibility that allows us to select block eigenvalues which are easy to compute.

We modify the power method to find block eigenvalues. This method is as follows:

\begin{enumerate}
\item Choose two indices $i$ and $j$.\label{step:choose}
\item Start with an estimate of the block 2-eigenvector $Y_0$, with $(Y_0)_{\{i,j\}} = I_2$, where $(Y_0)_{\{i,j\}}$ is the submatrix of $Y_0$ formed by taking rows $i$ and $j$. Set $k = 0$.
\item Calculate $AY_k$.\label{step:parpm}
\item Set $L_k = (Y_0)_{\{i,j\}}$.
\item Set $Y_{k+1} = A Y_k L_k^{-1}$.
\item Set $k = k+1$.
\item If $k$ does not exceed some fixed value, return to step \ref{step:parpm}.
\item Apply Gram-Schmidt orthogonalization to the columns of $Y_k$.\label{step:gs}
\item Set $L_k = Y^T_k A Y_k$.\label{step:rayl}
\item $L_k$ and $Y_k$ are the estimates for the block 2-eigenvalue and block 2-eigenvector respectively.
\end{enumerate}

In step \ref{step:choose}, $i$ and $j$ are chosen to coincide with a $2 \times 2$ block of the leading power of $A_l(a)$. Steps \ref{step:gs} and \ref{step:rayl} ensure that the columns of the approximate block eigenvector are orthonormal.

The following theorem justifies this method. Its proof is an extended version of Theorem \ref{thm:power} and will not be shown.
\begin{theorem}
Let $A$ be a matrix of power series satisfying the conditions of Theorem \ref{thm:power}. Let $Y = \left[\begin{array}{c|c} \y_1 & \y_2 \end{array}\right]$, where $\y_1$ and $\y_2$ are linear combinations of $\x_1$ and $\x_2$ such that $\y_1, \y_2 = O(1)$ and for some indices $i$ and $j$, $Y_{\{i,j\}} = I_2$. Then $Y$ is a block 2-eigenvector of $A$ with block 2-eigenvalue $L$, say. Suppose we have an estimate $\Yh$ of $Y$ which also has $\Yh_{\{i,j\}} = I_2$ and is accurate to $m$ terms, i.e.
\[\Yh - Y = O(z^m).\]

Then
\[\Lh = (A \Yh)_{\{i,j\}} = L + O(z^{l_1+m})\]
and
\[A \Yh \Lh^{-1} = Y + O(z^{m+l_3-l_2}).\]
\end{theorem}

This theorem shows that if the 2nd and 3rd largest eigenvalues of $A_l(a)$ are non-degenerate to leading power, using the modified power method with block 2-eigenvalues gives us a block 2-eigenvalue and associated 2-eigenvector. Therefore, if the first non-degeneracy occurs between the $k$th and ($k+1$)th largest eigenvalues, we will use block $k$-eigenvalues.

Once we have found the block eigenvalues, we must also deflate the matrix. The following theorem gives us the relevant formula.

\begin{theorem}
Let $A$ be a symmetric matrix with block 2-eigenvalue $L$ and corresponding 2-eigenvector $Y$. Suppose that $Y^TY = I$, and that $L$ has eigenvalues $\l_1$ and $\l_2$ (which are also eigenvalues of $A$). Then
\[A - YLY^T\]
has the same eigenvectors as $A$, and with the same corresponding eigenvalues, except that $\l_1$ and $\l_2$ are replaced by 0.
\end{theorem}
\begin{proof}
From Theorem \ref{thm:block}, we know that the columns of $Y$ are spanned by two eigenvectors of $A$, say $\x_1$ and $\x_2$. Let $\x$ be a different eigenvector of $A$ with corresponding eigenvalue $\l$. Then
\[(A - YLY^T)\x = A \x - YLY^T\x = \l \x\]
so $\l$ is an eigenvalue of the deflated matrix. On the other hand,
\[(A - YLY^T)Y = A Y - Y L I = 0\]
so any eigenvalue associated with $Y$ is set to 0.

\end{proof}

\subsection{Model-specific adjustments}

Most of the adjustments we made to the method are applicable to any model. However, we did make some adjustments which are specific to the hard squares model. The most important arises from the fact that by definition, $F(1,1)$ must be the zero matrix. Furthermore, it is easy to see from Equation \ref{eq:al} that only the top left $n_1 \times n_1$ block of $A_l(1)$ is nonzero. This enables us to treat $A_l(1)$ as a $n_1 \times n_1$ matrix, which makes manipulation faster. We note that in the case where the sizes are equal ($n_1 = n_2$), this means that we are calculating and keeping all of the eigenvalues of $A_l(1)$.

This latter point did in fact trip us up somewhat: when we tried to calculate and keep an extra eigenvalue (which should be 0), we produced a series with very high leading power, and gibberish for the eigenvector. Naturally this led to chaos when we tried to reduce the other matrices and repeat!

\section{Convergence}\label{sec:convergence}

If we use Equation \ref{eq:kappa} as written to calculate $\kappa$, the number of series terms we obtain is given by the largest eigenvalue of the $A_l$ matrices that we leave out in the shrinking step. More precisely, if the largest missing eigenvalue has leading power $z^a$, then the first term that is wrong in the approximation of $\kappa$ is $z^{4a}$. This is because all terms in Equation \ref{eq:kappa} involve 4th powers of the $A$ matrices. Table \ref{tab:eigen} shows the number of terms that we would produce if one matrix was limited in size and the other was unlimited.

\begin{table}
\begin{tabular}{lr|lr}
Size of $A(0)$ & Number of terms & Size of $A(1)$ & Number of terms \\
\hline
1 & 8 & 1 & 17 \\
2 & 16 & 2 & 25 \\
3 & 24 & 3 & 33 \\
4-5 & 32 & 4 & 41 \\
6-7 & 40 & 5 & 45 \\
8-10 & 48 & 6 & 49 \\
11-13 & 56 & 7-8 & 57 \\
14-17 & 64 & 9-11 & 65 \\
18 & 68 & 12-14 & 73 \\
19-22 & 72 & 15-18 & 81 \\
23-28 & 80 & 19 & 85 \\
29 & 88 & 20 & 89 \\
\end{tabular}\caption{Number of correct series terms from each matrix.}\label{tab:eigen}
\end{table}

However, as we remarked after Equation \ref{eq:kappa}, in practice we use $A_l$ and $F_l$ in place of $A$ and $F$ in this equation to calculate $\kappa$. It is now much less obvious how many terms we now derive. We calculated the number of correct terms at each size by comparing series resulting from small sizes with known terms from larger sizes. The results are in Table \ref{tab:terms}.

\begin{table}
\begin{tabular}{llrr}
& & Number of terms & Number of terms \\
Size of $A(0)$ & Size of $A(1)$ & from small matrices & from large matrices \\
\hline
2 & 2 & 16 & 20 \\
3 & 3 & 24 & 28 \\
4 & 4 & 32 & 36 \\
5 & 5 & 32 & 38 \\
6 & 6 & 40 & 44 \\
7 & 7 & 40 & 46 \\
8 & 8 & 48 & 52 \\
9 & 9 & 48 & 52 \\
10 & 10 & 48 & 54 \\
11 & 11 & 56 & 60 \\
12 & 12 & 56 & 60 \\
13 & 13 & 56 & 62 \\
14 & 14 & 64 & 68 \\
15 & 15 & 64 & 68 \\
16-17 & 15 & 64 & 70 \\
18 & 15 & 68 & 76 \\
19-21 & 15 & 72 & 76 \\
22 & 15 & 72 & 78 \\
\end{tabular}\caption{Number of correct series terms using small and large matrices.}\label{tab:terms}
\end{table}

It is apparent that if we set a desired number of terms, then set the matrices to the smallest size able to derive this number of terms, then using large matrices gives us 4 extra terms. Furthermore, if we use the largest size possible, using large matrices gives us 6 extra terms if this size is different. The only exception to this rule occurs at size $(n_1,n_2) = (18,15)$, which can be considered a special case in the sense that every other size produces a number of terms which is a multiple of 8.

Using this information, we ran the CTMRG method for sizes up to $(n_1,n_2) = (29,20)$, and were able to determine that this yields 92 terms for the partition function per site.

One question of interest is whether the method is indeed an $O(\alpha^{\sqrt{n}})$ method, where $n$ is now the number of terms derived. For this to be true, we would need the matrix size (which we denote by $m$ for this argument) to also grow like $O(\beta^{\sqrt{n}})$. Empirical evidence does indeed suggest that this relationship holds. Figure \ref{fig:conv} shows a plot of the logarithm of the matrix size against the square root of the number of terms, and it can be seen that a linear relationship is quite strongly apparent.

\begin{figure}\center
\includegraphics[scale=0.7]{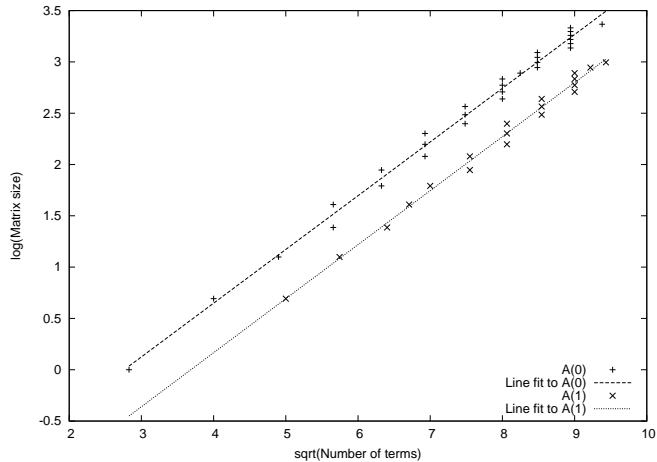}\caption{Matrix size vs. number of terms.}\label{fig:conv}
\end{figure}

The line fits have very close to the same slope, and it seems reasonable to conjecture that the true slopes are indeed identical. Taking the maximum of the two fitted slopes gives us $\beta \approx 1.69$. Now the CTMRG method is theoretically an $O(m^3)$ method if the shifted inverse power method is not used, and $O(m^4)$ otherwise, although in practice the latter case is more efficient. This gives us an estimate of $\alpha \approx 4.85$ for the former case and $\alpha \approx 8.21$ for the latter.

\section{Analysis}\label{sec:analysis}

It must be noted that the focus of this paper is on the CTMRG method used to generate the hard squares series, rather than results obtained from analysing the resulting series. Nevertheless, the series is much longer than anything previously generated, so we analyse it to see what results can be obtained.

The series which we have generated is the low-density series, where all spins are 0 in the base state. This series was generated to 43 terms by Baxter \emph{et al.} in \cite{CTM:Hsq}, but was not used for analysis, as they preferred the high-density series (the expansion in the variable $z^{-1}$), which was also generated via CTM. Unfortunately, we have so far been unable to generate this series with CTMRG.

Two points of interest in the hard squares model are the critical point at $z_c \approx 3.80$, and the dominant singularity at $z \approx -0.12$. The critical point is of interest because it is the transition point where one sublattice becomes preferentially occupied, i.e. the model changes from low- to high-density. However, because the dominant singularity for the low-density series is much closer to 0, this tends to `drown out' information about the critical point, so it is easier to analyse the high-density series for information about this point. Indeed, to our knowledge the low-density series has not been used to analyse the critical point since 1965 (\cite{CTM:Gaunt+Fisher}).

Our series is of sufficient length that we can make a reasonably accurate determination of the critical point. To do this, we analysed the magnetisation series
\[M(z) = \frac{d}{dz} \ln \kappa\]
using the method of differential approximants (\cite{CTM:DiffApp}). In short, this method fits a function to the series which satisfies a low-order differential equation with polynomial coefficients, then looks at the singularity structure of the fitted function. The critical exponent of the magnetisation series is $1 - \alpha$, where $\alpha$ is the specific heat exponent. Using homogeneous second-order approximants, we found
\[z_c = 3.79635(9), \alpha = 0.0020(17),\]
where the numbers in brackets are twice the standard deviation of the approximant estimates (though we note that this should not be taken as strict error bounds). These numbers are in line with the commonly held view that this model belongs to the Ising universality class, where $\alpha = 0$. They are also consistent with, though considerably less accurate than, the best estimates attained by high-density analysis (see for example \cite{CTM:Universality}).

A much more accurate determination can be made of the dominant unphysical singularity, also using differential approximants. This point is primarily of interest because it is known (\cite{CTM:Dhar}) that $M(-z)$ is the generating function of directed animals on the body-centred cubic (b.c.c.) lattice. Again using second-order approximants, we derived (where $\gamma$ is the negative exponent of the singularity)
\[z = 0.1193388818(6), \gamma = 0.171(14).\]
In addition, we can then use these estimates to estimate the critical amplitude by solving for various $n$ the equation
\[c_n = A (1/z)^n n^{\gamma-1}\]
and then plotting our results against $1/n$. This gives
\[A = 0.145,\]
though we would hesitate to give an error for this estimate. This results in a formula for the asymptotic growth of the number of directed b.c.c. animals as
\[c_n \sim 0.145 \times 8.379^n n^{-0.829}.\]

\section{Conclusion}\label{sec:conclusion}

In this paper, we have given a large amount of detail as to how to adapt the CTMRG method to derive series expansions, illustrating by calculating 92 terms of the hard squares partition function per site. A number of technical difficulties have been overcome, notably with the use of block eigenvalues.

It is clear that the method is very efficient for calculating series, and we are fairly confident in saying that it appears to be an $O(\alpha^{\sqrt{n}})$ method, at least for hard squares. If this is true, theoretically this represents a vast improvement over all other non-CTM based methods, which are exponential-time.

The CTMRG is by nature a very general method, theoretically applicable to any IRF model. Although in practice certain symmetry requirements are also necessary, we are certain that its scope is not limited to the hard squares model, and believe that it can be successfully applied to many models to derive series.

In particular, although the original formulation of the CTMRG is for spin models, we are currently engaged in adapting it to vertex and bond models, where the `spin' values lie on the bonds of the lattice. If this is successful, this would open up a whole new category of models which we can apply this method to.

\section*{Acknowledgements}

I would like to acknowledge Andrew Rechnitzer and Tony Guttmann for many helpful discussions, proof-reading and support, and MASCOS (Australia) and ANR (France) for their funding.

\appendix

\section{The hard squares partition function}

{\scriptsize
\thispagestyle{empty}
\begin{tabular}{lr}
$n$ & Coefficient of $z^n$ \\
\hline
0 & 	1 \\
1 &  1 \\
2 &  -2 \\
3 &  8 \\
4 &  -40 \\
5 &  225 \\
6 &  -1362 \\
7 &  8670 \\
8 &  -57253 \\
9 &  388802 \\
10 & -2699202 \\
11 & 19076006 \\
12 & -136815282 \\
13 & 993465248 \\
14 & -7290310954 \\
15 & 53986385102 \\
16 & -402957351939 \\
17 & 3028690564108 \\
18 & -22904845414630 \\
19 & 174175863324830 \\
20 & -1331044586131594 \\
21 & 10217222223168657 \\
22 & -78746146809812974 \\
23 & 609153211886323748 \\
24 & -4728123941310119629 \\
25 & 36812657530897835053 \\
26 & -287439461791025474818 \\
27 & 2250314840062625743472 \\
28 & -17660572072127314002800 \\
29 & 138917347311377551474338 \\
30 & -1095044102004611782219794 \\
31 & 8649079543673381406386578 \\
32 & -68441069128808194161922385 \\
33 & 542528768962390004584576547 \\
34 & -4307673277782673209498570830 \\
35 & 34255913017196256622645849406 \\
36 & -272811973711116137449858922289 \\
37 & 2175663718003877171512666515965 \\
38 & -17373555504340949646557187291612 \\
39 & 138907228460715779361866368091340 \\
40 & -1111918671840441187102586337375728 \\
41 & 8910623138600432871714003453719826 \\
42 & -71483639721296620300995136065253668 \\
43 & 574046483511726716038779843196291148 \\
44 & -4614334493396507062886044646610429157 \\
45 & 37125630601616372118866371971750653400 \\
46 & -298967336036118129407418784080218615722 \\
47 & 2409585254960886275025134297727824908884 \\
48 & -19436320799533420112481773783042415629261 \\
49 & 156900573920162022290578129250670083086420 \\
50 & -1267535404869110564352411619174739265523868 \\
51 & 10247264495520354226129600890924108669994610 \\
52 & -82900281399057902108151873229350373092209619 \\
53 & 671108737638240713935447505566478276195162695 \\
54 & -5436355713698184882980887629568939260561976810 \\
55 & 44064694805175046617091994106334371203302938776 \\
56 & -357381432990462157109211753457839726558288436818 \\
57 & 2900163105083745845730874102786889605714746995342 \\
58 & -23547968769944310985596027746076562077618218334718 \\
59 & 191300842560899308725071542860775532916665485120184 \\
60 & -1554908876921181737147789458998481638358161928897591 \\
61 & 12644749797865196843555951541679953304124755960793511 \\
62 & -102878765260913100430709623218183178373090721871184044 \\
63 & 837422752069874189432173744615001497684269956054808622 \\
64 & -6819631352357213787230910331148600271798284788605855567 \\
65 & 55560748326532902393539847760004905789472334071550252840
\end{tabular}

\begin{tabular}{lr}
$n$ & Coefficient of $z^n$ \\
\hline
66 & -452856402854393303217557453525610814293298860266644333788 \\
67 & 3692603110457838500058306311991565464177314769422347258712 \\
68 & -30121676657343800172725754276632905777873884902459021355361 \\
69 & 245807413474990027756549514844805809793849952994063754626736 \\
70 & -2006666888019302099239255549634083582058527531632344348155382 \\
71 & 16387602447917907925939835398242666596145404236646122967314476 \\
72 & -133878529017269645222962908444675311373581968623046805553003420 \\
73 & 1094101349658594729612361555133529988788157463539921142705674406 \\
74 & -8944399559315399008297389300844640091529020258243919220987964946 \\
75 & 73145551008354781012338946915052673817209937064486329418176787110 \\
76 & -598361927886673095890838492285516107595411125515110286142445348366 \\
77 & 4896386322427075789832417441388661081060656557690205676073687691148 \\
78 & -40079251232688073115885503442096697172216907602724038749014964900102 \\
79 & 328165082254688810946984293011935816745304828942186642808616061608366 \\
80 & -2687761676918736119845722240766742417562922822567399217372718456286685 \\
81 & 22019713136269272910821806308408983339580368950205667984685424553851833 \\
82 & -180448000438033885658918793091063759646692376290415150089214756908058884 \\
83 & 1479139461040169224120227477170760620707156578769244001285571905105338718 \\
84 & -12127744073180764329413943869671758523793164836282542085455077710510119460 \\
85 & 99463122933872289081188517872014743137876065832109978317420162733094536685 \\
86 & -815929689380612324987741425248826191973143825540604664462510726155651011716 \\
87 & 6694982697997270547775373671241966492741914685730214539163542377702398519634 \\
88 & -54947741307877958845071184518809752964246445865955568127226073015990501845122 \\
89 & 451077911529567493786366696817726882849261112979605658136036057774968937144332 \\
90 & -3703841147518398777454506425673146945950304284333846865189410967618798579127030 \\
91 & 30419357562872572130630138383935901849475451065777199789899946292436298406543592
\end{tabular}
}

\bibliographystyle{abbrv}
\bibliography{hardsq}

\end{document}